\documentclass[twoside]{article}

\usepackage[T1]{fontenc}
\usepackage[english]{babel}
\usepackage{amsthm,amsmath}
\RequirePackage[numbers]{natbib}
\RequirePackage[colorlinks,citecolor=blue,urlcolor=blue]{hyperref}
\usepackage{mathtools,amssymb}
\usepackage{graphicx,enumitem}
\graphicspath{{figures/}}
\usepackage[toc,page]{appendix}
\usepackage{verbatim}
\usepackage{moreverb}
\usepackage{bbm}
\usepackage{lmodern} 
\usepackage{framed}
\usepackage{float}
\usepackage{epstopdf}
\usepackage{authblk}

\newcommand{\defi}{\overset{\Delta}{=}}
\newcommand{\E}{\mathbb{E}}

\newcommand{\N}{\mathbb{N}}
\newcommand{\R}{\mathbb{R}}
\newcommand{\C}{\mathbb{C}}
\newcommand{\1}{\mathbf{1}}

\newcommand{\esp}[1]{\E\left[#1\right]}

\newcommand{\espt}[1]{\E_\theta\left[#1\right]}

\newcommand{\indii}[1]{\mathbbm{1}_{#1}}
\def\rmd{\mathrm{d}} 
\def\rme{\mathrm{e}} 
\def\rmi{\mathrm{i}} 
\newtheorem{prop}{Proposition}[section]
\newtheorem{theorem}{Theorem}[section]
\newtheorem{coro}{Corollary}[section]
\newtheorem{lemma}{Lemma}[section]

\newtheorem{definition}{Definition}[section]
\theoremstyle{definition}
\newtheorem{remark}{Remark}[section]
\newenvironment{hyp}[1]{
\begin{enumerate}[label=(\textbf{\sf #1}-\arabic*),resume=hyp#1]\begin{sf}}
{\end{sf}\end{enumerate}}

\begin{document}

\title{Nonparametric estimation of mark's distribution of an exponential Shot-noise process}

\author[$\dagger$,$\ddagger$]{Paul Ilhe \thanks{This research was partially supported by Labex DigiCosme (project ANR-11-LABEX-0045-DIGICOSME) operated by ANR as part of the program ``Investissement d'Avenir'' Idex Paris-Saclay (ANR-11-IDEX-0003-02).}}
\author[**]{\'Eric Moulines }
\author[$\dagger$]{Fran\c cois Roueff }
\author[$\ddagger$]{Antoine Souloumiac }
\affil[$\dagger$]{Telecom ParisTech, 46 rue Barrault, 75634 Paris Cedex 13, France}
\affil[$\ddagger$]{CEA, LIST, 91191 Gif-sur-Yvette Cedex, France}
\affil[**]{Ecole Polytechnique, Centre de Math\'ematiques Appliqu\'ees- UMR 7641, Route de Saclay 91128 Palaiseau}
\maketitle

\begin{abstract}
  In this paper, we consider a nonlinear inverse problem occuring in nuclear
  science. Gamma rays randomly hit a semiconductor detector which produces an
  impulse response of electric current. Because the sampling period of the
  measured current is larger than the mean interarrival time of photons, the
  impulse responses associated to different gamma rays can overlap: this
  phenomenon is known as \textit{pileup}. In this work, it is assumed that the
  impulse response is an exponentially decaying function. We propose a novel
  method to infer the distribution of gamma photon energies from the indirect
  measurements obtained from the detector. This technique is based on a formula
  linking the characteristic function of the photon density to a function
  involving the characteristic function and its derivative of the
  observations. We establish that our estimator converges to the mark density
  in uniform norm at a polynomial rate.  A limited Monte-Carlo experiment is
  provided to support our findings.
\end{abstract}



\begin{sloppypar}

\section{Introduction}

In this paper, we consider a nonlinear inverse problem arising in nuclear
science: neutron transport or gamma spectroscopy. For the latter, a radioactive
source, for instance an excited nucleus, randomly emits gamma photons according
to a homogeneous Poisson point process. These high frequency radiations can be
associated to high energy photons which interact with matter via three
phenomena : the photoelectric absorption, the Compton scattering and the pair
production (further details can be found in \cite{knoll1989radiation}). When
photons interact with the semiconductor detector (usually High-Purity Germanium
(HPGe) detectors) arranged between two electrodes, a number of electron-holes
pairs proportional to the photon transferred energy is created. Accordingly,
the electrodes generate an electric current called impulse response whenever
the detector is hit by a particle, with an amplitude corresponding to the
transferred energy. In this context, a feature of interest is the distribution
of this energy. Indeed, it can be compared to known spectra in order to
identify the composition of the nuclear source. In practice, the electric
current is not continuously observed but the sampling rate is typically smaller
than the mean inter-arrival time of two photons. Therefore, there is a high
probability that several photons are emitted between two measurements so that
the energy deposited is superimposed in the detector, a phenomenon called
pile-up. Because of the pile-up, it is impossible to establish a one-to-one
correspondence between a gamma ray and the associated deposited energy.

This inverse problem can be modeled as follows. The electric current generated in the detector
is given by a stationary shot-noise process \(\textbf{X}= \left(X_t \right)_{t  \in\R} \) defined by:
\begin{equation}
  \label{eq:def-Xt}
   X_t = \sum_{k : T_k \leq t} Y_k \, h(t-T_k) \;,
\end{equation}
where $h$ is the (causal) impulse response of the detector and
\begin{hyp}{SN}
\item\label{ass:poisson} $\sum_{k} \delta_{T_k, Y_k}$ is a Poisson
  point process with times $T_k \in \R$ arriving homogeneously with intensity $\lambda>0$ and independent
  i.i.d. marks $Y_k \in\R$ having a probability density function (p.d.f.)
  $\theta$ and cumulative distribution function (c.d.f) $F$.
\end{hyp}
We wish to estimate the density $\theta$ from a regular observation sample
$X_1,\dots,X_n$ of the shot noise \eqref{eq:def-Xt}. Note that the sampling
rate is set to 1 without meaningful loss of generality. If a different sampling
rate is used, e.g. we observe $X_\delta,\dots, X_{n\delta}$ for some $\delta
\neq 1$, it amounts to change $\lambda$ and to scale $h$ accordingly.

The process~(\ref{eq:def-Xt}) is well defined whenever the following condition
holds on the impulse response $h$ and the density $\theta$
\begin{equation}
  \label{eq:cns-shot}
  \int \min(1, |y\,h(s)|)\theta(y)\rmd y\rmd s < \infty\;.
\end{equation}

As shown in \cite{iksanov2003shot}, this condition is also necessary. Moreover,
the marginal distribution of $\bf{X}$ belongs to the class of infinitely
divisible (ID) distributions and has L\'{e}vy measure $\nu$ satisfying, for all
Borel sets $B$ in
$\R\setminus\{0\}$, 
\begin{align}
\label{eq:levytripletSN}
&\nu (B ) \defi \lambda \int_0^\infty \mathbb{P}\left(h(s)Y_1 \in B\right)\;
\rmd s \;. 
\end{align}

The ID property of the marginal distribution shows that
this estimation problem is closely related to the estimation of the L{\'e}vy
measure $\nu$. This property strongly suggests to
use estimators of the L\'{e}vy triplet, see for instance
\cite{neumann2009nonparametric} and
\cite{gugushvili2009nonparametric}. However, up to our best knowledge, these
estimators use the increments of the corresponding L\'evy process which are
i.i.d. and they assume a finite L\'evy Khintchine
measure. In contrast, the observations are not independent and the L\'{e}vy
measure of the process is infinite since from \eqref{eq:levytripletSN}, we
have that
\begin{equation} \nu(\R) = \lambda\int_0^\infty \mathbb{P}\left(h(s)Y_1 \in
    \R\right)\rmd s =\infty. \end{equation} In order to tackle this estimation
problem, we then propose to bypass the estimation of $\nu$ and directly
retrieve the density $\theta$ of the marks distribution $F$ from the empirical
characteristic function of the measurements. Coarsely speaking,
using~(\ref{eq:levytripletSN}), the L\'{e}vy-Khintchine representation provides
an expression of the characteristic function $\varphi_X$ of the marginal
distribution as a functional of $\theta$. The estimator is built upon replacing
$\varphi_X$ by its empirical version and inverting the mapping
$\theta\mapsto\varphi_X$. A more standard marginal-based approach would be to
rely  on the p.d.f. of $\mathbf{X}$. However,  the density of $X_0$ is intractable, which precludes the use of a
likelihood inference method. Consequently, although shot-noise models are
widespread in applications (for example, such models were used to model internet
traffic \cite{barakat2003modeling}, river streamflows
\cite{claps2005advances}, spikes in neuroscience
(\cite{holden1976models},\cite{hohn2001shot}) and in signal processing
(\cite{sequeira1995intensity} , \cite{sequeira1997blind})), theoretical results
on the statistical inference of shot-noise appear to be limited. Recently, Xiao
and al. (\cite{xiao2006inference}) provide consistent and asymptotically normal
estimators for parametric shot-noise processes with specific impulse responses.

In this contribution, we consider the particular case given by the following assumption. 
\begin{hyp}{SN} 
\item \label{cond:IRexpo}The impulse response $h$ is an exponential function with
  decreasing rate $\alpha>0$~:
\begin{equation*}
h(t) \defi \rme^{-\alpha t}\;\mathbbm{1}_{\R_+}(t) \;.
\end{equation*}
\end{hyp}
Under~\ref{cond:IRexpo}, the process $(X_t)_{t\in\R}$ is usually called an \emph{exponential
  shot-noise}. In this case, Condition~(\ref{eq:cns-shot}) becomes
\begin{equation}
\label{eq:cond-expo-shot-noise}
\esp{\log_+(\left|Y_1\right|)} <\infty\;. 
\end{equation}
Under~\ref{cond:IRexpo}, the process $(X_t)_{t\geq0}$ can alternatively be
introduced by considering the following stochastic differential equation
(SDE)~:
\begin{equation}
\label{eq:L{\'e}vyOU}
\rmd X_t = -\alpha X_t \rmd t + \rmd L_t \; , \quad X_0=x \in \R
\end{equation}
where $\textbf{L}=(L_t)_{t\geq 0}$ is a L\'evy process defined as the compound
Poisson process  
\begin{equation}\label{eq:CPP}
L_t \defi
  \sum_{k=1}^{N_\lambda(t)}Y_k\quad\text{with}\quad N_\lambda(t) \defi \sum_k\indii{T_k\leq t}\;,
\end{equation} 
where $(T_k,Y_k)_{k\geq 0}$ satisfies \ref{ass:poisson}. The solution to the
equation \eqref{eq:L{\'e}vyOU} is called a Ornstein-Uhlenbeck(O-U) process
(\cite[Chapter~17]{Sato-1999levy}) driven by \textbf{L} with initial condition
$X_0=x$ and rate $\alpha$.  Note that $\mathbf{L}$ defined by~(\ref{eq:CPP})
has L{\'e}vy measure $\lambda F$. Thus, by \cite[Theorem 17.5]{Sato-1999levy},
this Markov process admits a unique stationary version
if~(\ref{eq:cond-expo-shot-noise}) holds, and this stationary solution
corresponds to the shot-noise process \eqref{eq:def-Xt}.
\newline

In recent works, \cite[Brockwell, Schlemm]{brockwell2013parametric} exploit the integrated version of  \eqref{eq:L{\'e}vyOU} to recover the L\'evy process \textbf{L} and show that the increments of \textbf{L} can be represented as: \begin{equation*}
L_{nh}-L_{(n-1)h} = X_{nh}-X_{(n-1)h} +\alpha\int_{(n-1)h}^{nh} X_s \rmd s.
\end{equation*}
These quantities are only well estimated for high frequency observations so
that we cannot rely on this method in our regular sampling scheme.

To the best of our knowledge, the paper that best fits our setting is \cite{jongbloed2005nonparametric}. The authors propose a nonparametric estimation procedure from a low frequency sample of a stationary O-U process  which exploits the self decomposability property of the marginal distribution. The authors construct an estimator of the so called canonical function $k$ defined by:\begin{equation*}
\nu(\rmd x) =\frac{k(x)}{x}\rmd x .
\end{equation*} The two main additional assumptions are that $k$ is decreasing on $(0,\infty)$ and $\nu$ satisfies the integrability condition $\int_0^\infty (1\wedge x)\nu(\rmd x) <\infty$. In our setting (i.e. when specifying the L\'evy process to be the compound Poisson process defined in \eqref{eq:CPP}), it is easily shown that these conditions hold and the canonical function and the cumulative distribution of the marks are related by the equation:\begin{equation*}
k(x)=\lambda \mathbb{P}\left(Y_0 > x\right)=\lambda \left(1-F(x)\right).
\end{equation*}

In this article, we introduce an estimator of $\theta$ based on the empirical
characteristic function and a Fourier inversion formula. This algorithm is
numerically efficient, being able to handle large datasets typically used in
high-energy physics. Secondly, we establish an upper bound of the rate of
convergence of our estimator which is uniform over a smoothness class of
functions for the density $\theta$.

The paper is organized as follows. In Section~\ref{sec:main-result}, we introduce some
preliminaries on the characteristic function of an exponential shot-noise
process and provide both the inversion formula and the estimator of the density $\theta$. In particular, we derive an upper bound of convergence
for our estimator over a broad class of densities under the assumption that
\(\lambda/\alpha \) is known. In Section \ref{section:num_res}, we present in details the algorithm
used to perform the density estimation and illustrate our findings with a limited
Monte-Carlo experiment. Section \ref{section:err_bounds} provides error bounds for the empirical
characteristic function based on discrete-time observations and exploit the $\beta$-mixing structure of the process. Finally, Section \ref{Proofs} is devoted to the proofs of the various theorems.

\section{Main result}
\label{sec:main-result}
\subsection{Inversion formula}

As mentioned in the introduction, it is difficult to derive the probability
density function of the stationary shot-noise unless the marks are distributed
according to an exponential random variable and the impulse response is an
exponential function. In this case, it turns out that the marginal distribution
of the shot-noise is Gamma-distributed (the reader can refer to
\cite{bondesson1992shot} for details). In all other cases, we can only compute
the characteristic function of the marginal distribution of the stationary
version of the shot-noise when treating it as a filtered point process (see for
example \cite{resnick2007extreme} for details). We have for every real
$u$: \begin{equation}\label{SNcf} \varphi(u) \defi \E[\rme^{\rmi uX_0}]
  = \exp\left(\lambda  \int_{\R} \int_{0}^\infty(\rme^{\rmi uyh(v)}-1) \rmd
    v\,F(\rmd y)\right) .
\end{equation}

From \eqref{SNcf}, the characteristic function of $X_0$ can be expressed as follows:\begin{equation}\label{eq:functionalMellin}
\varphi_{X_0}(u) = \exp\left( \int_\R \lambda K_h(uy)F(\rmd y) \right) .
\end{equation}
where $K_h$, the kernel associated to $h$ is given by:\begin{equation*}
K_h(x) \defi \int_0^{+\infty} (\rme^{\rmi xh(v)}-1)\rmd v .
\end{equation*}
Note that if $h$ is integrable, then $K_h$ is well defined since, for any real $x$, $\int_0^\infty \vert \rme^{\rmi  xh(s)} -1 \vert \rmd s \leq |x|\int_0^\infty |h(s)|\rmd s $.
Moreover, if $h$ is integrable, then $K_h$ is a $\mathcal{C}^1(\R , \C)$ function whose derivative is bounded and equal to: \begin{equation*}
 K_h'(x) =\int_0^{+\infty} \rmi h(s)\rme^{\rmi xh(v)} \rmd v .
\end{equation*}
Furthermore, if $\esp{|Y_0|} < \infty$, then the characteristic function of $X_0$ is differentiable and we have: \begin{equation}
\label{SNcfderivative}
\varphi_{X_0}^\prime(u) =\lambda \; \varphi_{X_0}(u)  \, \int_\R yK_h'(uy)F(\rmd y) .
\end{equation}

Under~\ref{cond:IRexpo}, the kernel $K_h$ takes the form 
\begin{equation}\label{eq:K_alpha}
K_{h}(u)= \int_0^\infty \left(\rme^{\rmi u \rme^{-\alpha v}} -1 \right)\rmd v =\int_0^u \frac{\rme^{\rmi s}-1}{\alpha s} \rmd s  .
\end{equation}
With \eqref{SNcfderivative}, we obtain that
\begin{equation} \label{eq:characYcharacX}
\varphi_{X_0}'(u)=\varphi_{X_0}(u) \frac{\lambda}{\alpha u} \left( \varphi_{Y_0}(u) -1 \right)  .
\end{equation}

Since the marginal distribution of $X$ is infinitely divisible, we have by
\cite[Lemma 7.5.]{Sato-1999levy} that  $\varphi_X(u)$ does not vanish. If in
addition $\varphi_Y$ is integrable, \eqref{eq:characYcharacX} provides a way to
recover $\theta$ knowing $\alpha /\lambda$, namely, for all $x \in \R$,
 \begin{equation}
 \label{eq:inversion_densite}
 \theta(x) =\frac{1}{2\pi}\int_\R \rme^{-\rmi xu} \varphi_{Y_0}(u) du=\frac{1}{2\pi}\int_\R \rme^{-\rmi xu} \left( 1 +\frac{\alpha u}{\lambda} \frac{\varphi_{X_0}'(u)}{\varphi_{X_0}(u)} \right) \rmd u \;.
 \end{equation}
This relation shows that the estimation problem of the p.d.f. $\theta$ is directly related to the estimation of the second characteristic function. 
\begin{remark}
  We assume in the following that the ratio $\alpha / \lambda$ appearing in the
  inversion formula \eqref{eq:inversion_densite} is a known constant, as it
  typically depends on the measurement device. Interestingly, however, an
  estimator of this constant can be derived from \cite[Theorem
  1]{iksanov2003shot}, where it is shown that the marginal distribution $G$ of
  the stationary shot-noise is regularly varying at $0$ with index $\lambda /
  \alpha$, i.e. :
\begin{equation*}
G(x) \sim x^{\lambda / \alpha}L(x) \quad \text{,}\quad x \to 0 \;
\end{equation*}
with $L$ being slowly varying at $0$. Hence it is possible to estimate
$\alpha/\lambda$ by applying Hill's estimator \cite{hill1975simple} to
the sample $X^{-1}_1,\cdots,X^{-1}_n$. 
\end{remark}

\subsection{Nonparametric estimation}

Let $\hat{\varphi}_n(u) \defi n^{-1}\sum_{j=1}^n \rme^{\rmi uX_j}$ denotes the
empirical characteristic function (e.c.f.) obtained from the observations and
$\hat{\varphi}_n'$ its derivative. From
\eqref{eq:inversion_densite}, we are tempted to plug the e.c.f. of the
observations to estimate the p.d.f. $\theta$. Let $(h_n)_{n \geq 0}$
and$(\kappa_n)_{n \geq 0}$ be two sequences of positive numbers such that
$$
\lim_{n \to \infty} h_n =\lim_{n \to \infty} \kappa_n  = 0\;,
$$ and consider the
following sequence of estimators:
\begin{equation}
 \label{ESTIM}
 \hat{\theta}_n(x) \defi \max\left(\frac{1}{2\pi}\int_{-\frac{1}{h_n}}^{\frac{1}{h_n}} \rme^{-\rmi xu} \left(1 + \frac{\alpha u}{\lambda}\frac{ \hat {\varphi}'_n(u)}{\hat {\varphi}_n(u)}\mathbbm{1}_{|\hat {\varphi}_n(u)| > \kappa_n} \right)\rmd u  , 0 \right)\;.
 \end{equation}
 \begin{remark}\label{remark:MainTh}
   We estimate $1/\varphi(u)$ by \(
   \mathbbm{1}_{\lbrace\left|\hat{\varphi}_n(u)\right| \geq
     \kappa_n\rbrace}/\hat{\varphi}_n(u)\) with a suitable choice of a sequence
   \((\kappa_n)_{n \geq 1} \) which converges to zero. The constant $\kappa_n$
   is chosen such that $\left| \hat \varphi_n(u) -\varphi(u)\right|$ remains
   smaller than $\left|\hat \varphi_n(u)\right|$ and $\left|\varphi(u)\right|$
   with high probability in order to avoid large errors when inverting $\hat
   \varphi_n(u)$. In \cite{neumann2009nonparametric}, the authors deal with
   the empirical characteristic function of i.i.d. random variables. In this case, the
   deviations of $\sqrt{n}(\hat{\varphi}_n(u)-\varphi(u))$ are bounded in
   probability, hence, they use $1_{\lbrace|\hat{\varphi}_n(u)| \geq \kappa
     n^{-1/2}\rbrace}/\hat{\varphi}_n(u)$ as an estimator of
   $1/\varphi(u)$. Here we truncate the interval of integration \( \R \)
   by \( \left[ -h_n^{-1} , h_n^{-1} \right] \), where $h_n$ is a bandwidth
   parameter. This allows us to bound the estimation error $\hat
   \theta_n-\theta$ in sup norm. The deviation of
   $\sqrt{n}(\hat{\varphi}_n(u)-\varphi(u))$ on $\left[ -h_n^{-1} , h_n^{-1}
   \right]$ depends on $h_n$, see Theorem~\ref{Thm:ECF}. The resulting
   $\kappa_n$ is then taken slightly larger than $n^{-1/2}$.
 \end{remark}

 In order to evaluate the convergence rate of our estimator, we consider
 particular smoothness classes for the density $\theta$. Namely we define, for
 any positive constants $K, L, m$ and $s >1/2$,
\begin{align}
  \label{eq:smoothness-class}
 \Theta(K,L,s,m)=\Big\lbrace \theta \text{ is a density s.t. }&\int
 |y|^{4+m}\theta(y)\,\rmd y\leq K \;,\\
 &\nonumber\int_\R
 (1+|u|^2)^{s}\left|\mathcal{F}\theta(u)\right|^2\,\rmd u\leq L^2\Big\rbrace\;,
\end{align} 
where $\mathcal{F}\theta$ denotes the Fourier transform of $\theta$
$$
\mathcal{F}\theta(u)=\int \theta(y)\,\rme^{-\rmi y \,u}\;\rmd y \quad,\quad u\in\R \;.
$$
Hence $L$ is an upper bound of the Sobolev semi-norm of $\theta$. 
Note also that under \ref{ass:poisson}-\ref{cond:IRexpo}, $\theta$ belongs to 
$\Theta(K,L,s,m)$ is equivalent to assuming that
$$
\esp{\left|Y_0\right|^{4+m}} \leq K\quad\text{and}\quad\int_\R (1+|u|^2)^{s} |\varphi_{Y_0}(u) |^2 du \leq L^2 \;.
$$
In the following, under assumptions \ref{ass:poisson}-\ref{cond:IRexpo}, we use
the notation $\mathbb{P}_\theta$ and $\mathbb{E}_\theta$, where the subscript
$\theta$ added to the expectation and probability symbols indicates explicitly
the dependence on the unknown density $\theta$.  
The following result provides a bound of the risk $\mathbb{P}_\theta(
\|\theta-\hat{\theta}_n\|_\infty > M_n)$ for well chosen sequences
$(h_n)$, $(\kappa_n)$ and $(M_n)$, which is uniform over the densities $\theta\in
\Theta(K,L,s,m)$.

 \begin{theorem}
\label{theorem22}
Assume that the process $\textbf{X}=\left(X_t\right)_{t\geq 0}$ given by
\eqref{eq:def-Xt} satisfies the assumptions
\ref{ass:poisson}-\ref{cond:IRexpo} for some positive
constants $\lambda$ and $\alpha$. Let $K, L, m$ be positive
constants and $s>1/2$. Let $C_{K,L,m,\lambda/\alpha}$ be the constant defined by
$$
C_{K,L,m,\lambda/\alpha}\defi \exp\left(-\frac{\lambda}{\alpha}(K^{1/(4+m)}+L)\right)\;,
$$
and $C$ be a positive constant such that  $
 0<C<C_{K,L,m,\lambda/\alpha}\;.$
Set
\begin{equation*}
h_n= n^{-1/(2s+1+2\lambda/\alpha)} \quad\text{ and } \quad \kappa_n=C\left(1+h_n^{-1}\right)^{-2\lambda/\alpha} \;,
\end{equation*}
and define $\hat\theta_n$ by~(\ref{ESTIM}).
Then, for $n\geq 3$, the density estimator $\hat \theta_n$ satisfies
\begin{equation}\label{eq:main_th}
 \sup_{\theta\in\Theta(K,L,s,m)} \espt{ \left\|\theta-\hat{\theta}_n\right\|_\infty } \leq M n^{-(2s-1)/(4s+2+4\lambda/\alpha)}\log(n)^{1/2}\;,
\end{equation}
where $M>0$ is a constant only depending on $C, K,L,m,s$ and $\lambda/\alpha$.
\end{theorem}
\begin{proof}
See Section \ref{sec:gamma-choosen}.
\end{proof}
\begin{remark}
The constant $C$ in Theorem \ref{theorem22} might be adaptively chosen. Indeed, the well
known relationship (see \cite{daley1988introduction}[Chapter 6] for example) between the cumulant function of a filtered Poisson process and its intensity measure implies that the mean $\mu_\theta$ of $X_0$ is given by
\begin{equation}
\label{def:VarianceShot}
\mu_\theta= \lambda\espt{Y_0}\int_0^\infty \rme^{-\alpha s}\rmd s =\frac{\lambda}{\alpha}\espt{Y_0} \leq \frac{\lambda}{\alpha}K^{1/(4+m)} \;.
\end{equation}
Since $\mathbf X$ is ergodic (see Section \ref{section:err_bounds}),  the empirical mean  $\hat \mu_{n}$  of the sample
$X_1,\cdots,X_n$ converges to $\mu_\theta$ almost surely and thus, replacing $C$ in Theorem \ref{theorem22} by $\hat C_n \defi \exp(-\hat \mu_n)/2$, leads to the same rate of convergence since
$$
\lim_{n\to \infty} \hat C_n = \exp\left(-\frac{\lambda}{\alpha}\espt{Y_0}\right)/2 \in (0,C_{K,L,m,\lambda/\alpha}) \quad\text{ a.s.}
$$
\end{remark}
\begin{remark}
This theorem provides that the error in uniform norm converges at least at a polynomial rate $n^{-(2s-1)/(4s+2+8\lambda/\alpha)}\log(n)^{1/2}$ that depends both on the quantity $\lambda/\alpha$ and the smoothness coefficient $s$. For a given ratio $\lambda/\alpha$, the convergence rate becomes faster as $s$ increases and tend to behave as $n^{-1/2}$ when $s\to\infty$. On the other side, for a given smoothness parameter $s$, the rate of convergence decreases when the ratio $\lambda/\alpha$ tends to infinity. This can be interpreted as the consequence of the \textit{pileup} effect that occurs whenever the intensity $\lambda$ is large or the impulse response coefficient $\alpha$ is close to zero.
\end{remark}

\begin{remark}
  Based on the previous theorem, one might wonder whether the rates of
  convergence are optimal. According to similar but not identical problems
  (\cite{neumann2009nonparametric},\cite{gugushvili2009nonparametric}) in which
  authors estimate in a nonparametric fashion a L\'evy triplet (with finite
  activity) based on a low frequency sample of the associated process, the
  optimal rates of convergence are identical to ours. Our estimation procedure
  lies on stationary but dependent infinitely divisible random variables
  associated to an infinite L\'evy measure so that these results do not apply
  here. However we believe
  that the rates obtained in Theorem \ref{theorem22} are also optimal in this
  dependent context. The proof of this conjecture is left for future work.
\end{remark}

\section{Experimental results}\label{section:num_res}

The estimation procedure based on the estimator $\hat \theta_n$ given by \eqref{ESTIM} can be made time-efficient and thus well suited to a very large dataset. In nuclear applications, it is usual to deal with several million of observations while the intensity of the time-arrival point process can reach several thousand of occurrences per second. Typically, the shot-noise process in nuclear applications corresponding to the electric current is discretely observed for three minutes at a sampling rate of 10Mhz and the mean number of arrivals between two observations lies between 10 and 100. Such large values for the intensity and the number of sampled points motivate us to present a practical way to compute the estimator~(\ref{ESTIM}).

\subsection*{Practical computation of the estimator} 
In Section \ref{sec:main-result}, we have defined the estimator of mark's density by \eqref{ESTIM}. Although it theoretically converges to the true density of shot-noise marks, the evaluation of the empirical characteristic function and its derivative based on observations $X_1,\cdots,X_n$ might be time-consuming when the sample size $n$ is large. To circumvent this issue, we propose to compute the empirical characteristic function using the fast fourier transform of an appropriate histogram of the vector $X_1,\cdots,X_n$. More precisely, for a strictly positive fixed $h$, we consider the grid $G = \lbrace hl : \lfloor \min_{k \leq n}(X_k)/h\rfloor \leq l \leq \lceil \max_{k \leq n}(X_k)/h\rceil \rbrace $ and compute the normalized histogram $H$ of the sample sequence $(X_l)_{1\leq l\leq n}$ with respect to the grid $G$ defined by
\begin{equation*}
H(l)=\frac{1}{n}\sum_{k=1}^n 1_{ [G(l) ; G(l+1)]}(X_k)\text{,}\quad \lfloor\min_{k \leq n}(X_k)/h\rfloor \leq l \leq \lceil \max_{k \leq n}(X_k)/h\rceil -1 \;.
\end{equation*}

Denoting $m_n \defi \lfloor\min_{k \leq n}(X_k)/h\rfloor$ and $M_n \defi \lceil \max_{k \leq n}(X_k)/h\rceil -1$, remark that for every real $u$, we have:
\begin{equation*}
\hat{\varphi}_n(u)= \frac{1}{n}\sum_{k=1}^n \rme^{\rmi uX_k} =\frac{1}{n}\sum_{k=1}^n\sum_{l=m_n}^{M_n} 1_{ [G(l) ; G(l+1)]}(X_k) \;\rme^{\rmi uX_k} \;.
\end{equation*}
Replacing $1_{ [G(l) ; G(l+1)]}(X_k) \;\rme^{\rmi uX_k}$ by $1_{ [G(l) ; G(l+1)]}(X_k) \;\rme^{\rmi uh(l+1/2)}$ for any real $u$, we get an approximation of the empirical characteristic function by defining\begin{equation}
\label{def:ECFhisto}
\hat{\varphi}_{h,n}(u) \defi \sum_{l=m_n}^{M_n} H(l)\rme^{\rmi uh(l+1/2)}\;.
\end{equation}
For any real $u$, we have the following upper bounds \begin{equation*}
\left|\hat{\varphi}_{h,n}(u)  - \hat{\varphi}_n(u) \right| \leq \frac{h}{2}|u| 
\end{equation*} and \begin{equation*}
\left|\hat{\varphi}'_{h,n}(u)  - \hat{\varphi}'_n(u) \right| \leq \frac{h}{2}\left( 1+|u|h\sum_{l=m_n}^{M_n}H(l)\;(l+1/2)\right) \quad \;,
\end{equation*} showing that the approximations are close to the true functions for small values of $h$ and $u$. From these empirical characteristic functions, we construct an estimator of the marks' characteristic function $\varphi_Y$ setting for any positive $u$:
\begin{equation*}
\hat{\varphi}_{Y,h,n}(u) \defi 1+ \frac{\alpha}{\lambda} u \frac{\hat{\varphi}'_{h,n}(u)}{\hat{\varphi}_{h,n}(u)} \mathbbm{1}_{\left|\hat \varphi_{h,n}(u)\right|>\kappa_n} \;.
\end{equation*}
The advantage of using $\hat{\varphi}_{h,n}$ is that $\hat \varphi_{h,n} (u)$ and $\hat \varphi_{h,n}'(u)$ can be evaluated on a regular grid using the fast Fourier Transform algorithm. 

The last step in the numerical computation of the estimator \eqref{ESTIM} consists in evaluating the quantity \begin{equation*}\int_0^{h_n^{-1}} \rme^{-\rmi xu}\hat{\varphi}_{Y,h,n}(u) \rmd u = \int_0^\infty \rme^{-\rmi xu}\hat{\varphi}_{Y,h,n}(u)1_{\left[0,h_n^{-1}\right]}(u) \rmd u\;.
\end{equation*}
Using the Inverse fast Fourier Transform, we  approximate the integral on a regular grid $x \in$ by a Riemann sum.
\subsection*{Numerical results}
We now illustrate the finite sample behavior of our estimator on a simulated data set when the marks (in keV energy units) density follows a Gaussian mixture $\sum_{i=1}^3 p_i\mathcal{N}_{\mu_i,\sigma_i^2}(x)$ with
\[p= \left[
0.3\text{ }0.5\text{ }0.2\right] \quad \text{,}  \quad \mu= \left[
4\text{ }12\text{ }22\right]\quad \text{,}  \quad \sigma= \left[
1\text{ }1\text{ }0.5\right]\;.\]
Furthermore, in order to fit with nuclear science applications, where detectors
have a time resolution of 10 Mhz, corresponding to a sampling time
$\Delta=10^{-7}$ seconds. The parameters of the experiment are set to
$\alpha=8.10^8$ , $\lambda=10^9$. Moreover, in nuclear spectrometry, the bandwidth $h_n$ is directly related to the known precision of the measuring instrument. For the following numerical experiment, we set it to 2.5 which is in range with the detector resolution as described in \cite{knoll1989radiation}{, Chapter 4}. Figure \ref{fig:simulshot} below shows a
simulated sample path of such a shot-noise with its associated marked point
process.
 \begin{figure}[H]
 \begin{center}
 \resizebox{120mm}{!}{\includegraphics{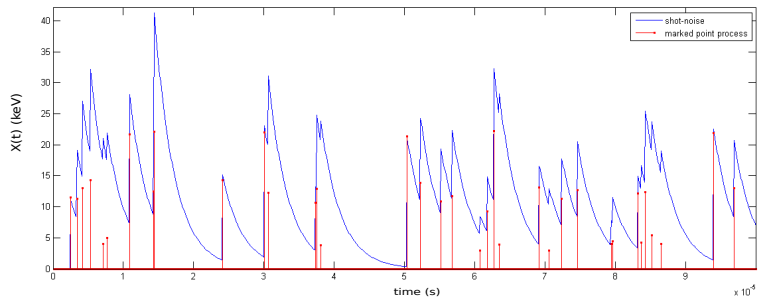}}
 \caption{Simulated Shot-Noise}\label{fig:simulshot}
 \end{center}
 \end{figure}
 As shown in Figure \ref{fig:GaussianEstim} below, our estimator $\hat
 \theta_n$ defined by \eqref{ESTIM} well retrieves the three modes of the
 Gaussian mixture as well as the corresponding variance from a sample of size
 $10^5$, which corresponds to a signal observed for one hundredth
 second. Current estimators used in nuclear spectrometry for similar data
 requires much longer measurements (up to 10 seconds). The reason is that these
 estimators do not consider observations where pile-up is suspected to occur, thus
 throwing away a large part of the available information.
\begin{figure}[H] 
\begin{center}
\resizebox{120mm}{!}{\includegraphics{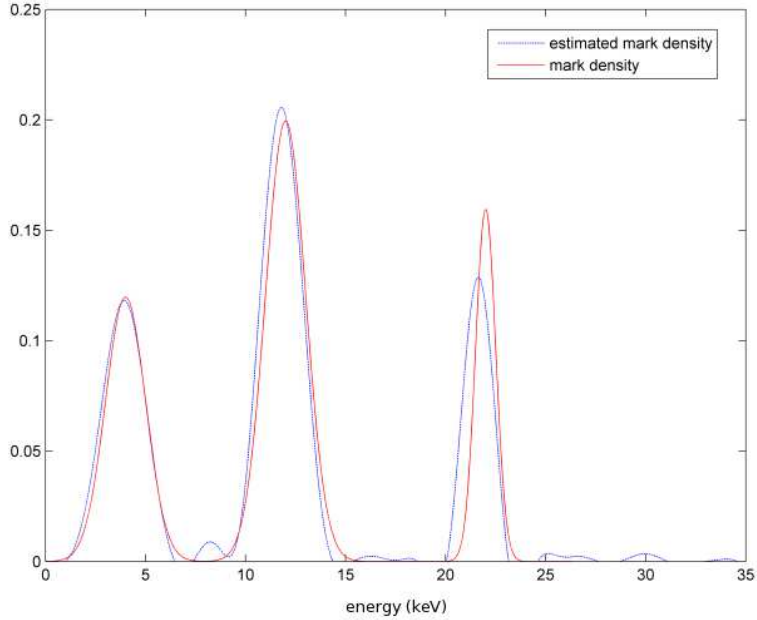}}
\caption{Gaussian mixture case} \label{fig:GaussianEstim}
\end{center}
\end{figure}
Moreover, an estimation of the risk $\espt{\|\theta-\hat{\theta}_n\|_\infty}$ is
provided in Table~\ref{tab:error} for the shot-noise configuration described at the beginning of this section and for three different sample sizes.
\begin{table}[H]
	\begin{center}
	\begin{tabular}{|c||c|c|}
		\hline 
		Sample size $(n)$ &  Mean $\ell^\infty$ error (variance) \\
		\hline
		$10^4$ & $0.1015$  ($5,10.10^{-4}$) \\ 
		\hline
		$10^5$ & $0.0741$ ($9,51.10^{-5}$) \\
		\hline
		$10^6$ & $0.0622 $ ($9,38.10^{-6}$) \\
		\hline
	\end{tabular}
\end{center}
	\caption{Mean  $\ell^\infty$ error  $\espt{\|\theta-\hat{\theta}_n\|_\infty}$ for
          different sample sizes $n$ with $\theta$ given by the Gaussian mixture
          defined above. The displayed values are obtained by averaging
          $\|\theta-\hat{\theta}_n\|_\infty$ over 100 Monte Carlo
          simulations. The variance of the error over the Monte Carlo runs is indicated between
          parentheses to assess the precision of the given estimates.} \label{tab:error} 
      \end{table}
      
\section{Error bounds for the empirical characteristic function and its derivatives}\label{section:err_bounds}

To derive Theorem~\ref{theorem22}, since our constructed estimator involves the
empirical characteristic function and its derivative, we rely on deviation
bounds for
\begin{equation} \label{ExpectationBounds}
\espt{\sup_{u \in \left[-h^{-1},h^{-1}\right]} \left| \hat{\varphi}_n^{(k)}(u)-\varphi^{(k)}(u)\right|} \quad , \quad k =0,1\quad , \quad h >0\;,
\end{equation} 
which are uniform over $\theta\in\Theta(K,L,s,m)$, where the smoothness class
$\Theta(K,L,s,m)$ is defined by \eqref{eq:smoothness-class}. Here, $\varphi^{(k)}$ and $\hat\varphi_n^{(k)} $ respectively denote the $k$-th derivative of the characteristic function and its empirical counterpart associated to the sample $X_1,\dots,X_n$. These bounds are of
independent interest and therefore are stated in this separate section.  Upper
bounds of the empirical characteristic function deviations have been derived in
the case of i.i.d. samples: \cite{gugushvili2009nonparametric}[Theorem 2.2.]
provides upper bounds of \eqref{ExpectationBounds} for i.i.d. infinitely
divisible random variables, based on general deviation bounds for the empirical
process of i.i.d. samples found in \cite{van1996weak}. Here we are concerned
with a dependent sample $X_1,\cdots,X_n$ and we rely instead on
\cite{doukhan1995invariance}. We obtain upper bounds with the same rate of
convergence as in the i.i.d. case but depending on the $\beta$-mixing
coefficients, see Theorem~\ref{Thm:ECF}. An additional difficulty in the
non-parametric setting that we consider is to derive upper bounds that are
uniform over smoothness classes for the density $\theta$, and thus to carefully
examine how the $\beta $ coefficients depend on $\theta$, see Theorem~\ref{thm:mixing:uniform-geometricbound}.

Let us first recall the definition of \( \beta \)-mixing coefficient (also
called absolutely regular or completely regular coefficient) as introduced by
Volkonskii and Rozanov \cite{volkonskii1959some}. For \(\mathcal{A},
\mathcal{B}\) two \(\sigma\)-algebras of $\Omega$, the coefficient \( \beta (\mathcal{A}, \mathcal{B} )  \) is defined by
\begin{equation*}
\beta (\mathcal{A}, \mathcal{B} ) \defi \frac{1}{2}\sup \sum_{(i,j) \in I\times J } \left| \mathbb{P}(A_i \cap B_j) - \mathbb{P}(A_i)\mathbb{P}(B_j) \right| 
\end{equation*}
the supremum being taken over all finite partitions \( (A_i)_{i \in I}\) and \(
(B_j)_{j \in J} \) of $\Omega$ respectively included in $\mathcal{A}$ and
$\mathcal{B}$. When dealing with a stochastic process \( (X_t)_{t \geq 0} \),
the \( \beta\)-mixing coefficient is defined for every positive \( s \) by:
\begin{equation*}
\beta(s) \defi \sup_{t \geq 0} \beta \left(\sigma\left( X_u , u \leq t \right), (\sigma\left( X_{s+u} , u \geq t \right) \right)
\end{equation*}
The process \( (X_t)_{t \geq 0} \) is said to be \( \beta\)-mixing if \(\lim_{
  t \to \infty} \beta(t) =0 \) and exponentially \( \beta \)-mixing if there
exists a strictly positive number \(a \) such that \(\beta(t)
= O(e^{-at}) \) as $t\to \infty$.

We first state a result essentially following from \cite{doukhan1995invariance} which specify how the $\beta$ coefficients allows us to derive bounds on the estimation of the characteristic function and its derivatives.

\begin{theorem}\label{Thm:ECF}
 Let $k$ be a non-negative integer and $X_1,\cdots,X_n$ a sample of a
 stationary $\beta$-mixing process. Suppose that there exists $C\geq1$ and $\rho\in(0,1)$ such that $\beta_n\leq C\rho^n$
for all $n\geq1$. Let $r>1$ and suppose
  that $\esp{|X_1|^{2(k+1)r}} <\infty$. 
  
Then there exists a constant $A$ only depending on $C,\rho$ and $r$ such that
for all $h>0$ and $n\geq1$, we have
\begin{multline}\label{ineq:MaximalIneqBetaMixing}
\esp{\sup_{u \in \left[-h^{-1},h^{-1}\right]} \left|
    \hat{\varphi}_n^{(k)}(u)-\varphi^{(k)}(u)\right|} \\ \leq
A\frac{\max\left(\esp{|X_1|^{2kr}}^{1/2r},\esp{|X_1|^{2(k+1)r}}^{1/2r}\right)\left(1+\sqrt{\log(1+h^{-1})}\right)}{n^{1/2}} \;.
\end{multline}

\end{theorem}
\begin{proof}
The proof is deferred to Section~\ref{sec-proof-empriical}
\end{proof}
It turns out that the stationary exponential shot-noise process $\textbf{X}$
defined by~\eqref{eq:def-Xt} is exponentially $\beta$-mixing if the
first absolute moment of the marks is finite, see
\cite{masuda2004multidimensional}[Theorem 4.3] for a slightly more general
condition. However, in order to obtain a uniform bound of the risk of our
estimator $\hat \theta_n$ over a smoothness class, a more precise result is
needed. In the sequel, we add a superscript $\theta$ to the $\beta$-mixing
sequence to make explicit the dependence with respect to the mark's density
$\theta$. The following theorem provides a geometric bound for the
$\beta$-mixing coefficients of the shot noise which is uniform over the class
$\Theta(K,L,s,m)$.

\begin{theorem}
\label{thm:mixing:uniform-geometricbound}
Let $X_1,\cdots,X_n$ be a sample of the stationary shot-noise process given by
\eqref{eq:def-Xt} satisfying \ref{ass:poisson}-\ref{cond:IRexpo}. Let $K,L,m>0$
and $s>1/2$.  Then there exist two constants $C>0$ and $\rho \in (0,1)$ only
depending on $\lambda$, $\alpha$, $K$, $L$, $s$ and $m$ such that, for all $n\geq1$,
\begin{equation}\label{Thm:UniformMixing}
 \sup_{\theta\in\Theta(K,L,s,m)}\beta_n^\theta \leq C\rho^n < \infty \;.
\end{equation}
\end{theorem}
\begin{proof}
See   Section \ref{section:app:beta-mixing}.
\end{proof}

 As a corollary of Theorems \ref{Thm:ECF} and \ref{thm:mixing:uniform-geometricbound}, we obtain error bounds for the empirical characteristic function when dealing with observations $X_1,\cdots,X_n$ of the stationary shot-noise process given by \eqref{eq:def-Xt}.
\begin{coro}\label{ECFbounds:Coro}
  Let $X_1,\cdots,X_n$ be a sample of the stationary shot-noise process given
  by \eqref{eq:def-Xt} satisfying \ref{ass:poisson}-\ref{cond:IRexpo}. Let
  $K,L,m>0$ and $s>1/2$ and let $k$ be an integer such that $0\leq k<1+n/2$. Then there exists a constant $B$ only depending on $k$,
  $\lambda$, $\alpha$, $K$, $L$, $s$ and $m$ such that for all and $n\geq1$,
\begin{align}
\nonumber\sup_{\theta \in \Theta(K,L,s,m)}\mathbb{E}_\theta \left[\sup_{u \in \left[-h^{-1},h^{-1}\right]} \left|
    \hat{\varphi}_n^{(k)}(u)-\varphi^{(k)}(u)\right|\right]
 \leq B\;\frac{1+\sqrt{\log(1+h^{-1})}}{n^{1/2}} \;.
\end{align} 

\end{coro}

This result can be compared to \cite{gugushvili2009nonparametric}[Theorem~2.2].
Note however that although our sample
has infinitely divisible marginal distributions, it is not independent and the
L\'evy measure is not integrable.

 \section{Proofs} \label{Proofs}

\subsection{Preliminary results on the exponential shot noise}
\label{sec:prel-results-expon}
We establish some geometric ergodicity results on the exponential shot noise
that will be needed in other proofs.

\begin{definition}[Geometric drift condition]
A Markov Kernel P satisfies a geometric drift condition (called $D(V,\mu,b)$) if there exists a measurable function $V : \R \to [1,\infty[$ and constants $(\mu,b) \in (0,1)\times \R_+$ such that \begin{equation*}
PV \leq \mu V + b \;.
\end{equation*}
\end{definition}

\begin{definition}[Doeblin set]
A set $C$ is called a $(m,\epsilon)$-Doeblin set if there exists a positive integer $m$ , a positive number $\epsilon$ and a probability measure $\nu$ on $\R$ such that, for any $x$ in $C$ and $A$ in $\mathcal{B}(\R)$
\begin{equation*}
P^m(x,A) \geq \epsilon \nu(A) \;.
\end{equation*}
\end{definition}
The following proposition is borrowed from \cite{tuominen1994subgeometric} and
relates explicitly the geometrical drift condition to the convergence in
$V$-norm (denoted by $|\cdot|_V$) to the stationary distribution.
\begin{prop}
\label{Prop:GeometricErgodicity}
Let $P$ be a Markov kernel satisfying the drift condition
$D(V,\mu,b)$. Assume moreover that for some $d > 2b(1-\mu)-1$ , $m \in
\N-\lbrace 0\rbrace$ and $\epsilon \in (0,1)$, the level set $ \lbrace V \leq d\rbrace$ is an
$(m,\epsilon)$-Doeblin set. Then $P$ admits a unique invariant measure $\pi$
and $P$ is $V$-geometrically ergodic, that is, for any $0<u<\epsilon/\left(b_m+\mu^md-1+\epsilon\right)_+\vee 1$ , $n \in \N$ and $x \in \R$,
\begin{equation}
\|P^n(x,\cdot)-\pi \|_V \leq c(u)[\pi(V) +V(x)]\rho^{\lfloor n/m\rfloor}(u)
\end{equation}
where
\begin{itemize}
\item  $b_m = \frac{b}{\min  V}\frac{1-\mu^m}{1-\mu}$
\item $c(u)=u^{-1}(1-u)+\mu^m+b_m$
\item $\rho(u) = \left(1-\epsilon+u(b_m+\mu^md+\epsilon-1)\right)\vee \left(1-u\frac{(1+d)(1-\mu^m-2b_m}{2(1-u)+u(1+d)}\right)$
\item $\pi(V)=\mathbb P(Z\in V)$ where $Z$ is a random variable with distribution $\pi$
\end{itemize}
\end{prop}
In order to apply such a result to the sample $(X_1,\cdots,X_n)$ of the
exponential shot-noise defined by \eqref{eq:def-Xt}, observe that it is the
sample of a Markov chain which satisfies 
the autoregression equation $ X_{i+1}=\rme^{-\alpha}X_i +W_{i+1} \;,$
where the sequence of innovations $(W_i)_{i\in\mathbb{Z}}$ is made up of i.i.d. random variables distributed as 
\begin{equation}\label{def:innovationSequence}
W_0 \defi \sum_{k=1}^{N_\lambda([0,1])} Y_k \rme^{-\alpha U_k} \;.
\end{equation}
where
\begin{itemize}
\item $N_\lambda([0,1])$ is a Poisson r.v. with mean $\lambda$,
\item $(Y_i)_{i \geq 1}$ are i.i.d. r.v.'s with probability density function $\theta$,
\item $(U_i)_{i \geq 1}$ are i.i.d. and uniformly distributed on $[0,1]$,
\item all these variables are independent.
\end{itemize}
In the following, we denote by  $Q_\theta$ the Markov
kernel associated to the Markov chain $(X_i)_{i\geq0}$ under \ref{ass:poisson}-\ref{cond:IRexpo}.


\begin{prop}[Uniform Geometric drift condition]\label{Proof:GeomDrift}
Let $K,L,m>0$ and $s>1/2$ and let $\theta\in\Theta(K,L,s,m)$. Then the Markov
kernel $Q_\theta$ satisfies the drift condition $D(V,\mu,b)$, where
\begin{equation}
\label{def:ProofMixing:GeomDrift}
V: x \to 1+ |x| \quad ,\quad \mu = \rme^{-\alpha}\quad,\quad b = 1+ \lambda K^{1/(4+m)}- \rme^{-\alpha} \;.
\end{equation}
\end{prop}
\begin{proof}We have for all $\theta\in\Theta(K,L,s,m)$ and $x\in\mathbb{R}$,
\begin{align}
\nonumber Q_\theta V(x)
&\nonumber =\espt{1+\left|\rme^{-\alpha}x+W_0\right|}\\ 
&\nonumber \leq \rme^{-\alpha}V(x) + 1-\rme^{-\alpha} +\lambda K^{1/(4+m)} =\mu V(x) +b \;.
\end{align}
\end{proof}
\begin{remark}\label{RemarkGeometricDrift}
A similar result holds for the functions $V_i:x \to 1+|x|^i$ where $i\in\lbrace 1,\cdots,\lfloor 4+m\rfloor\rbrace$.
\end{remark}
\begin{prop}[Doeblin set]
\label{Proof:DoeblinSet}
Let $l>1$, $K,L,m>0$ and $s>1/2$ and define $V$ as in~(\ref{def:ProofMixing:GeomDrift}). There exists  $\epsilon>0$
only depending on $l$, $\alpha$, $\lambda$, $K,L,m>0$ and $s$ such that, for all
$\theta \in \Theta(K,L,s,m)$, the Markov kernel $Q_\theta$ admits $\{V\leq l\}$
as an $(1,\epsilon)$-Doeblin set.
\end{prop}
\begin{proof}
Let $\theta \in \Theta(K,L,s,m)$.  Denote by  $\check{\theta}$ the density of random variable $Y_1\rme^{-\alpha U_1}$ with $U_1$ and $Y_1$
  two independent random variables respectively distributed uniformly on
  $[0,1]$  and with density $\theta$. It is easy to show that, for all $v\in\mathbb{R}$,
\begin{align}
\label{def:ThetaCheck}
 \check{\theta}(v) = \frac{1}{\alpha v} \int_v^{v\rme^{\alpha}} \theta(y)\rmd y \; .
\end{align}
The distribution of $W_0$ is thus given by the infinite mixture
\begin{equation}
\label{def:ProofMixing:density}
\rme^{-\lambda}\delta_0(\rmd \xi) + \sum_{k=1}^\infty \frac{\lambda^k }{k!}\rme ^{-\lambda}\check{\theta}^{*k}(\xi)\rmd \xi \defi \rme^{-\lambda}\delta_0(\rmd \xi) + (1-\rme^{-\lambda})\tilde{f}_\theta(\xi)\rmd \xi \;,
\end{equation}
where $\delta_0$ is the Dirac point mass at 0 and $\check{\theta}^{*k}$ denote
the $k$-th self-convolution of $\check \theta$. It follows that, for all Borel
set $A$,
\begin{align} \label{def:ProofMixing:MarkovKernel}
\nonumber Q_\theta(x,A)
 =\rme^{-\lambda}\1_A(\rme^{-\alpha}x)+\int_{A}(1-\rme^{-\lambda})\tilde f_\theta(\xi+\rme^{-\alpha}x)\;\rmd\xi\;.
\end{align}
In order to show that $\lbrace V \leq l\rbrace$ is a $(\epsilon,1)$-Doeblin-set for the
kernels $Q_\theta$, it is sufficient to exhibit a probability measure $\nu$
such that, for all $|x| \leq l-1$ and all Borel set $A$
\begin{equation*}
\int_{A} \tilde{f}_\theta(\xi +\rme^{-\alpha}x)\rmd \xi \geq (1-\rme ^{-\lambda})^{-1}\epsilon \;\nu\left(A\right) \;.
\end{equation*}
Hence if for each $\theta\in\Theta(K,L,s,m)$ we find $c(\theta)<d(\theta)$ such that 
\begin{equation*}
\epsilon'=\inf_{\theta\in\Theta(K,L,s,m)}\inf_{c(\theta)\leq\xi\leq d(\theta)}\inf_{|x|\leq l-1}[d(\theta)-c(\theta)]\,\tilde{f}_\theta(\xi +\rme^{-\alpha}x) >
0 \;,
\end{equation*}
 the result follows by taking $\nu$ with density $[d(\theta)-c(\theta)]^{-1}\1_{[c(\theta),d(\theta)]}$ and $\epsilon= (1-\rme ^{-\lambda})\epsilon'$.
By definition of $\tilde f_\theta$ above, it is now sufficient to show that there
exist $c(\theta)<d(\theta)$ and $k(\theta)\geq1$ such that
\begin{equation*}
\inf_{\theta\in\Theta(K,L,s,m)}\inf_{c(\theta)\leq\xi\leq d(\theta)}\inf_{|x|\leq l-1}[d(\theta)-c(\theta)]\,\check{\theta}^{*k(\theta)}(\xi +\rme^{-\alpha}x) >
0 \;.
\end{equation*}
Observe that for $c\leq\xi\leq d$ and $|x|\leq l-1$ we have $\xi
+\rme^{-\alpha}x\in[c-\rme^{-\alpha}(l-1),d+\rme^{-\alpha}(l-1)]$. So for any
interval $[c',d']$ of length $d'-c'>2\rme^{-\alpha}(l-1)$, we may set $c=c'+\rme^{-\alpha}(l-1)<d=d'-\rme^{-\alpha}(l-1)$
so that  $c\leq\xi\leq d$ and $|x|\leq l-1$ imply $\xi
+\rme^{-\alpha}x\in[c',d']$. 
Hence the proof boils down to showing that for each $\theta\in\Theta(K,L,s,m)$, there exist $c'(\theta)<d'(\theta)$ with
$d'(\theta)-c'(\theta)>2\rme^{-\alpha}(l-1)$ and $k(\theta)\geq1$ such that
\begin{equation}\label{eq:doeblin-to-show}
\inf_{\theta\in\Theta(K,L,s,m)}\inf_{c'(\theta)\leq\xi\leq d'(\theta)}[d'(\theta)-c'(\theta)-2\rme^{-\alpha}(l-1)]\,
\check{\theta}^{*k(\theta)}(\xi)
> 0 \;.
\end{equation}
By Lemma~\ref{lem:theta-check-min}, there exists $a>0$, $\Delta>0$, $\delta>1$ and
$\epsilon_0>0$ and such that  $\epsilon_0$ and $\Delta=b-a$ only depend on
$m$, $K$, $L$ and $s$ (although $a$ may depend on $\theta$), and
$$
\inf_{a<x<(a+\Delta)/\delta}\check\theta(x)>  \epsilon_0 \;.
$$
Finally, Lemma~\ref{Lemma:ConvolutionOfFunctions} and the previous
bound yield~(\ref{eq:doeblin-to-show}), which concludes the proof.
\end{proof} 
\subsection{Proof  of Theorem \ref{thm:mixing:uniform-geometricbound}}
\label{section:app:beta-mixing}
As explained in Section~\ref{sec:prel-results-expon}, $(X_i)_{i\geq0}$ is a
stationary $V$-geometrically ergodic Markov chain with Markov kernel denoted by
$Q_\theta$. By \cite{davydov1973mixing},
the $\beta$-coefficient of the stationary Markov chain $(X_i)_{i\geq0}$ can be
expressed for all $n\geq1$ and $\theta\in\Theta(K,L,s,m)$ as
$$
\beta^\theta_n =\int_\R \| Q^n_\theta(x,\cdot) -\pi_\theta  \|_{TV}\;\pi_\theta(dx) \;,
$$
where $\pi_\theta$ is the invariant marginal distribution and $|\cdot|_{TV}$
denotes the total variation norm, i.e. the $V$-norm with $V=1$.
Combining Propositions \ref{Proof:GeomDrift}, \ref{Proof:DoeblinSet} and~\ref{Prop:GeometricErgodicity},  we can
find constants $C>0$ and $\rho \in(0,1)$ only depending on $\lambda$, $\alpha$,
$K,L,s$ and $m$ such that
\begin{equation*}
\|Q_\theta^n(x,\cdot)-\pi_\theta\|_{TV} \leq \|Q_\theta^n(x,\cdot)-\pi_\theta\|_V \leq C\left( 2+\espt{|X_1|}+ |x| \right) \rho^n \;,
\end{equation*}
where $V(x)=1+|x|$. The last two displays yield
\begin{align}
\nonumber\beta_n^\theta
\nonumber&\leq 2C\;(1+ \espt{|X_1|})\rho^n\\
&\leq 2C\;(1+ \lambda K^{1/(4+m)}) \rho^n \;,
\end{align}
which concludes the proof.

\subsection{Proof of Theorem~\ref{ECFbounds:Coro}}
\label{sec-proof-empriical}

In \cite{doukhan1995invariance}, the authors establish a Donsker invariance principle for the process $\lbrace Z_n(f) , f \in\mathcal{F}\rbrace$ where $Z_n \defi n^{-1/2}\sum_{i=1}^n \left(\delta_{X_i} - P\right)$ is the normalized centered empirical process associated to a stationary sequence of $\beta$-mixing random variables $(X_1,\cdots,X_n)$ with marginal distribution $P$ and  $\mathcal{F}$ is a class of functions satisfying an entropy condition. To be more precise, suppose that the sequence $(X_i)_{i \geq 1}$ is $\beta$-mixing with $\sum_{n \in\N} \beta_n < \infty$. The mixing rate function $\beta$ is defined by $\beta(t)=\beta_{\lfloor t\rfloor}$ if $t\geq 1$ and $\beta(t)=0$ otherwise while its c\`adl\`ag inverse $\beta^{-1}$ is defined by:
\begin{equation*}
\beta^{-1}(u) \defi \inf_{ t \geq 0} \lbrace \beta(t) \leq u \rbrace
\end{equation*}
Further, for any complex-valued function $f$, denote by $Q_f$ the quantile function of the r.v. $|f(X_0)|$ and introduce the norm:
\begin{equation*}
\|f\|_{2,\beta} \defi \left( \int_0^1 \beta^{-1}(u)Q_f(u)^2\rmd u\right)^{1/2} \;.
\end{equation*}
The space $\mathcal{L}_{2,\beta}$ is defined as the class of functions $f$ such that $\|f\|_{2,\beta} < \infty$. In \cite{doukhan1995invariance}, the authors proved that $\left(\mathcal{L}_{2,\beta}, \|\cdot\|_{2,\beta}\right) $ is a normed subspace of $\mathcal{L}_2$. A useful and trivial result from the definition of the norm $\mathcal{L}_{2,\beta}(P)$ provides the following relation:
\begin{equation*}
|f| \leq |g| \Rightarrow \|f\|_{2,\beta} \leq \|g\|_{2,\beta}\;.
\end{equation*}
For any real $r >1$, another useful (less trivial) result in \cite{doukhan1995invariance} states that under the condition \begin{center}
$\sum_{n \geq 0} \beta_n\;n^{r/\left(r-1\right)}< \infty,$
\end{center}
we have $\mathcal{L}_{2r} \subset \mathcal{L}_{2,\beta}$ with the additional inequality
\begin{equation}
\label{ineq:BetaNormVSLpSpaces}
\|f\|_{2,\beta} \leq \|f\|_{2r}\sqrt{1+r\sum_{n \geq 0} \beta_n\;n^{r/\left(r-1\right)}} \;,
\end{equation}
where $\|f\|_{2r}=\esp{|f(X_0)|^{2r}}^{1/2r}$ denote the usual  $L^{2r}$-norm. 

Now, we can state a result directly adapted from \cite{doukhan1995invariance}[Theorem 3]  that will serve our goal to prove Theorem \ref{theorem22}. For the sake of self-consistency, we recall that, given a metric space of real-valued funtions $(E,\|\cdot\|)$ and two functions $l$, $u$ , the bracket $\left[l,u\right]$ represents the set of all functions $f$ such that $l\leq f \leq u$. For any positive $\epsilon$, $\left[l,u\right]$ is called an $\epsilon$-bracket if $\|l-u\|<\epsilon$.
\begin{theorem} \label{ThmDRM}
Suppose that the sequence $(X_i)_{i\geq1}$ is exponentially $\beta$-mixing and
that there exists $C\geq1$ and $\rho\in(0,1)$ such that $\beta_n\leq C\rho^n$
for all $n\geq1$. 
Let $\sigma>0$ and let $\mathcal{F} \subset \mathcal{L}_{2,\beta}$ be a class
of functions such that for every $f$ in $\mathcal{F}$,
$\|f\|_{2,\beta} \leq \sigma$.
Define 
$$
\phi(\sigma)=  \int_0^\sigma \sqrt{1+\log\left(N_{[\;]}\left(u,\mathcal{F},\|\cdot\|_{2,\beta}\right)\right)} \;\rmd u\, ,
$$
where $N_{[\;]}(u,\mathcal{F},\|\cdot\|_{2,\beta})$ denotes the bracketing number, that is, the minimal number of
$u$-brackets with respect to the norm $\|\cdot\|_{2,\beta}$ that has to be used
for covering $\mathcal{F}$. Suppose that the two following assumptions hold.
\begin{enumerate}[label=(DMR\arabic*),align=left]
\item\label{ThmDRM:EnvelopeCond} $\mathcal{F}$ has an envelope function $F$
  such that $\|F\|_{2r}<\infty$ for some $r>1$.
\item\label{ThmDRM:IntegrabilityCond} $\phi(1) <\infty \;.$ 
\end{enumerate}

Then there exist a constant $A>0$ only depending on $C$ and
$\rho$ such that, for all integer $n$, we have
\begin{equation}
\esp{\sup_{f\in\mathcal{F}}\left|Z_n(f)\right|} \leq A\phi(\sigma) \left(1+\frac{\|F\|_{2r}}{\sigma\sqrt{1-r^{-1}}}\right) \;.
\end{equation}

\end{theorem}

Having this result at hand, we now remark that \eqref{ExpectationBounds}, for a fixed integer $k$, can be rewritten as
\begin{equation}
\esp{\sup_{u \in {\left[-h^{-1},h^{-1}\right]}} \left| \hat\varphi_n^{(k)}(u)-\varphi^{(k)}(u)\right|} = n^{-1/2}\esp{\sup_{f \in \mathcal{F}^k_{h} }\left|Z_n(f)\right|} \, ,
\end{equation}
where
\begin{equation}
\mathcal{F}^{k}_{h} \defi \lbrace f_u : x \to (ix)^ke^{iux} , u \in \left[-h^{-1},h^{-1} \right]\rbrace .
\end{equation}
The proof of Theorem \ref{Thm:ECF} based on an application of the previous theorem is as follows.
\begin{proof}
We apply Theorem \ref{ThmDRM} for a fixed integer $k$,  $\mathcal{F}=\mathcal{F}^k_{h}$, $F=F_k$ and $r=(4+m)/4$ where $F_k:x\to |x|^k$. 
\small\subparagraph*{Assumption \ref{ThmDRM:EnvelopeCond}:}  

Let $k$ be a fixed integer. On the one hand, the function $F_k$ is an envelope function of the class $\mathcal{F}^{k}_{h}$ and on the other hand, for any real $r>1$, from \eqref{ineq:BetaNormVSLpSpaces}, we have
\begin{align} \label{BetaNormEnvelope}
\|F_k\|_{2,\beta} &\leq  \esp{|X_1|^{2kr}}^{1/2r}\sqrt{1+r\sum_{n \geq 0} \beta_n\;n^{r/\left(r-1\right)}}\defi \sigma_{r,k} < \infty  \;.
\end{align}

\small\subparagraph*{Assumption \ref{ThmDRM:IntegrabilityCond}:}

For $k$ a fixed integer, the class $\mathcal{F}^{k}_{h}$ is Lipschitz in the index parameter: indeed, we have for every $s,t$ in $ \left[-h^{-1},h^{-1} \right]$ and every real $x$
\begin{equation}
\left|(\rmi x)^k\rme^{\rmi sx}-(\rmi x)^k\rme^{\rmi tx}\right| \leq \left| s-t\right| \left|x\right|^{k+1}
\end{equation}
A direct application of \cite{van1996weak}[Theorem 2.7.11] for the classes $\mathcal{F}^k_{h}$ gives for any $\epsilon >0$:
\begin{align} \label{ProofThm33:1}
\nonumber N_{[\,]}\left(2\epsilon||F_{k+1}||_{2,\beta},\mathcal{F}^k_{h},||\cdot||_{2,\beta}\right) &\leq N\left(\epsilon,\left[-h^{-1},h^{-1} \right],|\cdot|\right) \\
& \leq 1+\frac{2h^{-1}}{\epsilon}
\end{align}
where $N$ and $N_{[\;]}$ are respectively called the covering numbers and bracketing number (these numbers respectively represent the minimum number of balls and brackets of a given size necessary to cover a space with respect to a given norm). 
From \eqref{ProofThm33:1}, it follows that for any $\sigma>0$, we have
\begin{align}
\nonumber\phi(\sigma) &=
\int_0^\sigma \sqrt{1+\log\left(N_{[\;]}\left(u,\mathcal{F}_h^k,\|\cdot\|_{2,\beta}\right)\right)}\rmd u\\
\label{ThmDRM:Conclusion} & \leq \int_0^\sigma\sqrt{1+\log\left(1+\frac{4||F_{k+1}||_{2,\beta}h^{-1}}{u}\right)}\rmd u \\
\nonumber & \leq \int_0^\sigma \left(1+\frac{2||F_{k+1}||_{2,\beta}^{1/2}h^{-1/2}}{u^{1/2}} \right)\rmd u \\
& =\sigma+4\sqrt{\sigma}\,\|F_{k+1}\|_{2,\beta}^{1/2}h^{-1/2}<\infty
\end{align}
because we supposed $F_{k+1} \in \mathcal{L}_{2r}$ and $\beta_n \leq C\rho^n$ which, from \eqref{ineq:BetaNormVSLpSpaces}, implies that $\|F_{k+1}\|_{2,\beta}<\infty$.
\subparagraph*{Conclusion of the proof}
The application of Theorem \ref{ThmDRM} gives
$$
\espt{\sup_{u \in {\left[-h^{-1},h^{-1}\right]}} \left| \hat{\varphi}_n^{(k)}(u)-\varphi^{(k)}(u)\right|} \leq \tilde A \frac{\phi(\sigma_{r,k})}{n^{1/2}} \;
$$
where $\tilde{A}=A(1+1/\sqrt{1-r^{-1}})$ since, from \eqref{BetaNormEnvelope}, we have $\|F_k\|_{2r}/\sigma_{r,k} \leq 1$.\newline
Set $c_{r,\bar\beta}\defi \sqrt{1+r\sum_{n \geq 0} \beta_n\;n^{r/\left(r-1\right)}} $. From  \eqref{ThmDRM:Conclusion} and \eqref{BetaNormEnvelope}, we can write
$$
\phi(\sigma) \leq \int_0^{\sigma} \sqrt{1+\log\left(1+\frac{4||F_{k+1}||_{2r}c_{r,\bar\beta}h^{-1}}{u}\right)}\rmd u \;,
$$
For $\sigma=\sigma_{r,k}$, we get after the change of variable $v= \frac{4||F_{k+1}||_{2,r}c_{r,\bar\beta}\sigma_{r,k}h^{-1}}{u}$
\begin{align*}
\phi(\sigma_{r,k})\leq  \max\left(\|F_k\|_{2r},\|F_{k+1}\|_{2r}\right)c_{r,\bar \beta}\left(1+h^{-1}\int_{h^{-1}
}^\infty \sqrt{\log(1+v)}\frac{\rmd v}{v^2}\right) \;.
\end{align*}

By Lemma \ref{Lemma:Karamata}, we get for a universal constant $B>0$ that
$$
\phi(\sigma_{r,k})\leq B\max\left(\|F_k\|_{2r} \|F_{k+1}\|_{2r}\right)c_{r,\bar \beta}\left(1+\sqrt{\log\left(1+h^{-1}\right)}\right)\;.
$$



In the particular context of Corollary \ref{ECFbounds:Coro}, we use the fact that $\|F_k\|_{2r}$ can be bounded by $\max(1,K^{4k/(4+m)})$ and $c_{r,\bar \beta}$ by a constant only depending on the parameters $K,L,s,m$
\end{proof}

\subsection{Proof of Theorem \ref{theorem22}}
\label{sec:gamma-choosen}
%

\begin{proof}

We denote by $\theta_n^0$ the function defined by :\begin{equation}
\label{eq:biais}
\theta_n^0(x) \defi \max\left(0,\frac{1}{2\pi}\int_{- h_n^{-1}}^{ h_n^{-1}}\rme^{-\rmi xu}\varphi_{Y_0}(u) \rmd u\right)\;.
\end{equation} 
Since $\theta \in \Theta(K,L,s,m)$ and $s>1/2$, we have that $\mathcal{F}\left[\theta\right]$ is integrable and, under $\theta \geq 0$, we have
\begin{equation}
\label{eq:inversion_formula}
\theta(x) = \max\left(0,\frac{1}{2\pi}\int_\R \rme^{-\rmi xu}\varphi_{Y_0}(u) \rmd u\right) \;.
\end{equation}
We decompose the error in infinite norm as
\begin{equation}
\label{Proof:DecompositionError}
\|\theta-\hat{\theta}_n\|_\infty  \leq \|\theta-\theta_n^0\|_\infty  + \|\theta_n^0-\hat{\theta}_n\|_\infty\;.
\end{equation}
From \eqref{eq:biais} and \eqref{eq:inversion_formula}, we get
\begin{align}
\nonumber\|\theta-\theta_n^0\|_\infty &\leq \frac{1}{\pi}\int_{ h_n^{-1}}^\infty \left|\varphi_{Y_0}(u)\right|\rmd u \\
&\nonumber =\frac{1}{\pi}\int_{ h_n^{-1}}^\infty \left|u^{-s}u^{s}\varphi_{Y_0}(u)\right|\rmd u \\
&\nonumber\leq \frac{1}{\pi} \left(\int_{ h_n^{-1}}^{\infty} \left|u\right|^{-2s}\rmd u\right)^{1/2}\left(\int_{ h_n^{-1}}^{\infty} \left|u^{s}\varphi_{Y_0}(u)\right|^2\rmd u\right)^{1/2}\\
&\label{proof:term_deterministic}\leq \frac{L}{\pi}\frac{ h_n^{s-1/2} }{2s-1} =\frac{L}{(2s-1)\pi} n^{-(2s-1)/(4s+2+4\lambda/\alpha)} \;.
\end{align}   
where we used the Cauchy-Schwartz inequality and the assumption that~$\theta \in\Theta(K,L,s,m)$.\newline
We conclude with a bound of the term involving $\|\theta_n^0-\hat{\theta}_n\|_\infty$  in~(\ref{Proof:DecompositionError}). To this end, the following inequality will be useful. Using \eqref{eq:characYcharacX} and the mean-value theorem, we have
\begin{equation}
\label{eq:maj_psi}
\sup_{u\in\R}\left|\frac{\varphi_{X_0}^\prime(u)}{\varphi_{X_0}(u)}\right| \leq \frac{\lambda}{\alpha}\sup_{u\in\R}\left|\varphi_{Y_0}^\prime(u)\right| \leq \frac{\lambda}{\alpha}\espt{\left|Y_0\right|}\leq \frac{\lambda}{\alpha}K^{1/(4+m)}\;.
\end{equation}

By \eqref{eq:characYcharacX}, we can bound the term $\|\theta_n^0-\hat{\theta}_n\|_\infty$ by
\begin{align*}
& \|\theta_n^0-\hat{\theta}_n\|_\infty =\frac{\alpha}{\lambda\pi}\left|\int_{-h_n^{-1}}^{h_n^{-1}}  \frac{\varphi_{X_0}'(u)}{\varphi_{X_0}(u)}-\frac{\hat{\varphi}'_n(u)}{\hat{\varphi}_n(u)}\indii{\left|\hat{\varphi}_n(u)\right| \geq \kappa_n} \rmd u\right|\\
&\leq \frac{\alpha}{\lambda\pi}\int_{-h_n^{-1}}^{h_n^{-1}} \left| \frac{\varphi_{X_0}'(u)}{\varphi_{X_0}(u)}-\frac{\hat{\varphi}'_n(u)}{\hat{\varphi}_n(u)}\indii{\left|\hat{\varphi}_n(u)\right| \geq \kappa_n}\right| \rmd u\\
&\leq \frac{2\alpha h_n^{-1}}{\lambda\pi}   \sup_{ \left|u\right| \leq h_n^{-1} } \left| \frac{\varphi_{X_0}'(u)}{\varphi_{X_0}(u)}-\frac{\hat{\varphi}'_n(u)}{\hat{\varphi}_n(u)}\indii{\left|\hat{\varphi}_n(u)\right| \geq \kappa_n}\right| \\
&\leq \frac{2\alpha h_n^{-1}}{\lambda\pi}\left(\sup_{ \left|u\right| \leq h_n^{-1} } \left| \frac{\varphi_{X_0}'(u)}{\varphi_{X_0}(u)}-\frac{\hat{\varphi}'_n(u)}{\hat{\varphi}_n(u)}\right|\indii{\left|\hat{\varphi}_n(u)\right| > \kappa_n} +\sup_{ \left|u\right| \leq h_n^{-1} }\left| \frac{\varphi_{X_0}'(u)}{\varphi_{X_0}(u)}\right|\indii{\left|\hat{\varphi}_n(u)\right| \leq \kappa_n} \right)\\
&  \defi A_{n,1} + A_{n,2}\;.
\end{align*}

Writing $\frac{\varphi_{X_0}^\prime}{\varphi_{X_0}}-\frac{\hat \varphi_{n}^\prime}{\hat \varphi_{n}}$ as $\left(\frac{\varphi_{X_0}^\prime}{\varphi_{X_0}}-\frac{\varphi_{X_0}^\prime}{\hat \varphi_{n}}\right)+\left( \frac{\varphi_{X_0}^\prime}{\hat \varphi_{n}}-\frac{\hat \varphi_{n}^\prime}{\hat \varphi_{n}} \right) $, the term $A_{n,1}$ can be bounded as follows.
\begin{align*}
&\sup_{ \left|u\right| \leq h_n^{-1}} \left| \frac{\varphi_{X_0}'(u)}{\varphi_{X_0}(u)}-\frac{\hat{\varphi}'_n(u)}{\hat{\varphi}_n(u)}\right|1_{\left|\hat{\varphi}_n(u)\right| > \kappa_n} \\&\leq \kappa_n^{-1}\sup_{ \left|u\right| \leq h_n^{-1}}\left|\Psi(u)\right|\left|\hat{\varphi}_n(u)-\varphi_{X_0}(u)\right| + \kappa_n^{-1}\sup_{ \left|u\right| \leq h_n^{-1}}\left|\hat{\varphi}'_n(u)-\varphi_{X_0}'(u)\right|
\end{align*}
Thus, using \eqref{eq:maj_psi}, we get
\begin{align*}
\espt{A_{n,1} }&\leq \frac{2  h_{n}^{-1} \kappa_n^{-1} K^{1/(4+m)} \espt{\sup_{ \left|u\right| \leq h_{n}^{-1}}\left|\hat{\varphi}'_n(u)-\varphi_{X_0}'(u)\right|}}{\pi} \\
&+ \frac{2\alpha h_{n}^{-1} \kappa_n^{-1} \espt{\sup_{ \left|u\right| \leq h_{n}^{-1}}\left|\hat{\varphi}_n(u)-\varphi_{X_0}(u)\right|}}{\pi\lambda} \;.
\end{align*}
The two terms on the right hand side can be bounded using Corollary \ref{ECFbounds:Coro} with $r=(4+m)/4$. It gives
\begin{equation*}
\espt{\sup_{\left|u\right| \leq h_{n}^{-1}}\left|\hat{\varphi}'_n(u)-\varphi_{X_0}'(u)\right|} \leq \frac{B\;K^{4/(4+m)}\left(1+\sqrt{\log(1+h_{n}^{-1})}\right)}{n^{1/2}}
\end{equation*}
and
\begin{equation*}
\espt{\sup_{\left|u\right| \leq h_{n}^{-1}}\left|\hat{\varphi}_n(u)-\varphi_{X_0}(u)\right|} \leq \frac{B\left(1+\sqrt{\log(1+h_{n}^{-1})}\right)}{n^{1/2}} \;.
\end{equation*}
In the following, for two positive quantities $P$ and $Q$, possibly depending
on $\theta$ and $n$ we use the notation   
\begin{equation}
  \label{eq:notation-lessssim}
  P \lesssim Q\Longleftrightarrow\text{for all $n\geq3$,}  \sup_{\theta \in \Theta(K,L,s,m)} \frac PQ <\infty \;.
\end{equation} 
($P$ is less than $Q$ up to a multiplicative constant uniform over $\theta \in \Theta(K,L,s,m)$). We thus have that
\begin{align}
\nonumber
 \espt{A_{n,1} }&\lesssim  \frac{1+\sqrt{\log(1+h_{n}^{-1})}}{\,\kappa_n\, n^{1/2}h_{n}}\\
 &\nonumber \lesssim \frac{1+\sqrt{\log(1+h_{n}^{-1})}}{\,\, n^{1/2}h_{n}^{\lambda/\alpha+1}}\\
 &\label{Proof:term_A1}
  \lesssim  n^{-(2s-1)/(4s+2+4\lambda/\alpha)}\log\left(n\right)^{1/2}\;,
\end{align}
where we used the fact that $\kappa_n^{-1}\leq 2C^{-1}h_n^{\lambda/\alpha}$ for any integer $n$.\newline
We now bound $A_{n,2}$. From \eqref{eq:maj_psi}, remark that
\begin{align}
\nonumber\espt{A_{n,2}}&\leq \frac{2 }{\pi h_n}K^{1/(4+m)}\mathbb{P}_\theta \left(\exists u \in\left[-h_n^{-1},h_n^{-1}\right] ,\left|\hat \varphi_n(u)\right|\leq \kappa_n \right)\\
&\label{proof_A2}\leq \frac{2 }{\pi h_n}K^{1/(4+m)}\mathbb{P}_\theta \left(\inf_{\left|u\right|  \leq h_{n}^{-1}} \left|\hat \varphi_n(u)\right| \leq   \kappa_n  \right)\;,
\end{align}
From Lemma \eqref{lemma:LemmaMainResult}, we have 
\begin{align}
\nonumber\inf_{\left|u\right|  \leq h_{n}^{-1}} \left|\hat \varphi_n(u)\right|&\geq \inf_{\left|u\right|  \leq h_{n}^{-1}} \left| \varphi_{X_0}(u)\right| -\sup_{\left|u\right|  \leq h_{n}^{-1}} \left|\hat \varphi_n(u)-\varphi_{X_0}(u)\right|\\
&\nonumber\geq C_{K,L,m,\lambda/\alpha}\left(1+h_n^{-1}\right)^{-\lambda/\alpha} -\sup_{\left|u\right|  \leq h_{n}^{-1}} \left|\hat \varphi_n(u)-\varphi_{X_0}(u)\right|\;.
\end{align}
It follows that
\begin{align*}
&\mathbb{P}_\theta \left(\inf_{\left|u\right|  \leq h_{n}^{-1}} \left|\hat \varphi_n(u)\right| \leq   \kappa_n  \right)\\
&\leq \mathbb{P}_\theta\left(\sup_{\left|u\right|  \leq h_{n}^{-1}} \left|\hat \varphi_n(u)-\varphi_{X_0}(u)\right| \geq (C_{K,L,m,\lambda/\alpha}-C)\left(1+h_n^{-1}\right)^{-\lambda/\alpha} \right)\;.
\end{align*}
 Since $0<C<C_{K,L,m,\lambda/\alpha}$, applying Corollary~\ref{ECFbounds:Coro} combined to the Markov's inequality, and using \eqref{proof_A2}, we get
\begin{align}\nonumber
\espt{A_{n,2}}& \lesssim 
\frac{1+\sqrt{\log\left(1+h_n^{-1}\right)}}{n^{1/2}\kappa_n h_n}\\
\label{Proof:Term2}
&\lesssim  n^{-(2s-1)/(4s+2+4\lambda/\alpha)}\log\left(n\right)^{1/2}\;,
\end{align}
Equations~(\ref{proof:term_deterministic}), \eqref{Proof:term_A1} and
\eqref{Proof:Term2} imply~(\ref{eq:main_th}) and the proof is concluded.
\end{proof}

\appendix
\section{Useful lemmas}
The following classical embedding will be useful.
\begin{lemma}[Sobolev embedding] \label{lemma:SobolCond} 
Let $K,L,m>0$ and $s
  >1/2$.  Let $\theta\in\Theta(K,L,s,m)$
  defined in \eqref{eq:smoothness-class}. Then, for any $\gamma
  \in (0,(s-1/2)\wedge 1)$, there is a constant $C>0$ depending on $L$, $s$ and
  $\gamma$ such that, for every real numbers $x,y$,
\begin{equation}
\left| \theta(x)-\theta(y) \right| \leq C \left| x-y\right|^{\gamma} \;,
\end{equation}
where 
\begin{equation}
  \label{eq:Cdef-embedding}
C =\frac{3}{2\pi}L\;\left(\int_\R\frac{|\xi|^{2\gamma}}{(1+|\xi|^2)^s}\rmd \xi\right)^{1/2}\;.  
\end{equation}
\end{lemma}
The following result is used in the proof of Proposition~\ref{Proof:DoeblinSet}.
\begin{lemma}
\label{ProofMixing:LemmaMinorizationTheta}
Let $K,L,m>0$ and $s >1/2$.  Let  $\gamma
  \in (0,(s-1/2)\wedge 1)$ and $\theta\in\Theta(K,L,s,m)$. 
Then, there exists $0<a\leq T_K$ such that
$$
\inf_{a\leq x\leq a+\Delta}\theta(x)\geq \frac1 {16} (2K)^{-1/(4+m)}\;,
$$
where $T_K=(2K)^{-1/(\gamma(4+m))}$, $\Delta=(2K)^{-1/(\gamma(4+m))}(16C)^{-1/\gamma}$
with $C$ defined by~(\ref{eq:Cdef-embedding}). 
\end{lemma}
\begin{proof}
We first show that, for every $T>0$,
we have
$$
\sup_{|x|\leq T} \theta(x) \geq (2T)^{-1}\left( 1-T^{-(4+m)}K\right) \;,
$$
Denote by $Y$ a random variable with p.d.f $\theta$ belonging to the class $\Theta(K,L,s,m)$. On the one side, we have $$
 \mathbb{P}\left( |Y| \leq T \right) \leq 2T \; \sup_{|x|\leq T} \theta(x)
 $$ and on the other side
$$
\mathbb{P}\left( |Y| \leq T \right) = 1 -\mathbb{P}\left( |Y| > T \right)\geq 1-\esp{|Y|^{4+m}}T^{-(4+m)} \geq \left( 1-T^{-(4+m)}K\right) \;,
$$ 
where the first inequality is obtained via an application of the Markov
inequality.
Setting $T_K = (2K)^{1/(4+m)}$, we thus have
$$
\sup_{|x|\leq T_K} \theta(x) \geq (4T_K)^{-1} \;.
$$

Moreover, since $\theta$ is continuous, we can without loss of generality suppose that there exists a positive number $a$  in the interval $(0,T_K]$ such that
$$
\theta(a) \geq(8T_K)^{-1} \;.
$$
From Lemma \ref{lemma:SobolCond}, there exists a positive number $\Delta=(16T_KC)^{-1/\gamma}$, independent of the choice of $\theta$ such that
$$
 \inf_{x\in[a,a+\Delta]} \theta(x) \geq (16T_K)^{-1} .
$$

\end{proof}

\begin{lemma}\label{lem:theta-check-min}
Let $K,L,m, \alpha>0$ and $s >1/2$.  Let  $\gamma
  \in (0,(s-1/2)\wedge 1)$ and $\theta\in\Theta(K,L,s,m)$. Define $T_K=(2K)^{-1/(4+m)}$, $\Delta=(2K)^{-1/(\gamma(4+m))}(16C)^{-1/\gamma}$
with $C$ defined by~(\ref{eq:Cdef-embedding}) and let $\delta$ be a positive number satisfying
  $$
  1< \delta< \min\left(\rme^\alpha , \frac{T_K+\Delta}{T_K} \right)\;.
  $$
  For any strictly positive $v$, define the function $\check \theta $ by
  $
  \check \theta (v)=\frac{1}{\alpha v}\int_v^{v\rme^\alpha} \theta(x)\rmd x\;.
  $
Then, there exists $0<a\leq T_K$ such that
$$
\inf_{a\leq v\leq (a+\Delta)/\delta}\check \theta(v)\geq \frac{(2K)^{-1/(4+m)}(\delta-1)} {16 \alpha} \;.
$$

\end{lemma}
\begin{proof}
From Lemma \ref{ProofMixing:LemmaMinorizationTheta}, we have
$$
\inf_{a\leq x\leq a+\Delta}\theta(x)\geq \frac1 {16} (2K)^{-1/(4+m)}\;,
$$
for some $a \in (0,T_K]$.
Let $ \delta \in (1, \rme^\alpha \wedge \frac{T_K+\Delta}{T_K}) \;.$
Since $(a+\Delta)/a$ is a decreasing function in $a$ for a fixed $\Delta$ and $0\leq a \leq T_K$, we have that
$$
(a+\Delta)/a \geq \frac{T_K+\Delta}{T_K}\;
$$
so that $ \delta < (a+\Delta)/a \;.$
For any $v \in [a,(a+\Delta)/\delta]$, we have
\begin{align*}
\check \theta(v) &=\frac{1}{\alpha v}\int_v^{v\rme^\alpha} \theta(x) \rmd x \geq \frac{1}{\alpha v}\int_v^{v\delta} \theta(x) \rmd x\\
& \geq \frac{v\delta-v}{\alpha v}\inf_{x\in[v,v\delta]} \theta(x)\\
&\geq \frac{\epsilon_K (\delta-1)}{\alpha}\;.
\end{align*}
which concludes the proof.
\end{proof}
The following elementary lemma generalizes the previous result for convolutions of lower bounded functions.
\begin{lemma}
\label{Lemma:ConvolutionOfFunctions}
Let $\theta$, $\tilde \theta$ two positive functions such that there exist positive numbers $a,b,c,d,\epsilon$ and $\tilde\epsilon  $ satisfying
\begin{equation*}
\theta( x) \geq \epsilon \mathbbm{1}_{[a,b]}(x) \quad \text{and} \quad \tilde\theta( x) \geq \tilde\epsilon\mathbbm{1}_{[c,d]}(x)
\end{equation*}
Then, for any $\delta$ satisfying $0<\delta < (b-a)\wedge (d-c)$, we have \begin{equation}
\left(\theta \star \tilde\theta \right) (x) \geq \min(1,\delta)\;\epsilon\tilde\epsilon\; \mathbbm{1}_{[a+c+\delta,b+d-\delta]}(x) \;.
\end{equation}
As a consequence, for any integer $n$ in $\N^*$, we have
\begin{equation}
\theta^{\star n} (x) \geq \left(\min\left(1,\frac{b-a}{2n}\right)\right)^{n-1}\epsilon^n\mathbbm{1}_{[na+(b-a)/2,nb-(b-a)/2]}(x) \;.
\end{equation}
\end{lemma}
A lower bound of the decay of the absolute value of the shot-noise characteristic function is given by the following lemma.
\begin{lemma} \label{lemma:LemmaMainResult}
	Assume that the process \textbf{X} is given by \eqref{eq:def-Xt}
	under~\ref{ass:poisson}-\ref{cond:IRexpo} with some positive constant $\alpha$
	and $\lambda$.  Let $K, L, m$ and $s$ be positive constants. Then for all
	$\theta\in\Theta(K,L,s,m)$ and $u\in\R$, we have 
	\begin{equation}
		\left| \varphi_{X_0}(u)\right| \geq C_{K,L,m,\lambda/\alpha}\left(1+|u|\right)^{-\lambda/\alpha}
	\end{equation}
	where $C_{K,L,m,\lambda/\alpha} \defi  \exp(-\lambda (L+K^{1/(4+m)})/\alpha)$.
\end{lemma}

\begin{proof}
	From \eqref{eq:functionalMellin} and \eqref{eq:K_alpha}, we have for
        all $u\in \R$,
\begin{equation*}
		\varphi_{X_0}(u)=\exp\left(\frac{\lambda}{\alpha} \int_{\R}\left(\int_0^{ux} \frac{\rme^{\rmi v}-1}{v}\rmd v\right)\theta(x)\rmd x\right)\;.
	\end{equation*}
	If follows that 
\begin{align}
          \nonumber\left|\varphi_{X_0}(u)\right|&=\exp\left(\frac{\lambda}{\alpha}\int_\R \left(\int_0^{ux}\frac{\cos(v)-1}{v} \rmd v\right) \theta(x) \rmd x\right)\\
          &\label{eq:PhiXmod}
=\exp\left(-\frac{\lambda}{\alpha}\int_0^{|u|}
            \frac{1-\text{Re}\left(\varphi_{Y_0}(z)\right)}{z}\rmd z\right)\;.
	\end{align}
First, we have for any real $z$ and any function $\theta \in \Theta(K,L,s,m)$,
\begin{align}
\nonumber \left|1-\text{Re}\left(\varphi_{Y_0}(z)\right)\right|&=\left|\int_{\R}\left(1-\rme^{\rmi xz}\right)\theta(x)\rmd x\right|\\
&\nonumber\leq \int_{\R}\left|1-\rme^{\rmi xz}\right|\theta(x)\rmd x \\
\nonumber &\leq 2\int_\R \left|\sin(xz/2)\right|\theta(x)\rmd x\\
&\nonumber\leq \int_{\R}\left|xz\right|\theta(x)\rmd x \leq K^{1/(4+m)}|z| \;.
\end{align}
We thus get that 
\begin{align}
\label{eq:upluspetitque1}
\int_0^1 \frac{1-\text{Re}\left(\varphi_{Y_0}(z)\right)}{z}\rmd z \leq
K^{1/(4+m)} \;.
\end{align}
Now, for $|u|\geq1$, we have that
\begin{align}
          \nonumber
\int_1^{|u|} \frac{1-\text{Re}\left(\varphi_{Y_0}(z)\right)}{z}\rmd z 
&\leq \log |u| +\int_1^{|u|}
z^{-1}\left|\text{Re}\left(\varphi_{Y_0}(z)\right)\right|\rmd z\\
\label{eq:u-plusgrand1}
& \leq \log |u| + L \;,
\end{align}
where we use the Cauchy-Schwartz inequality and 
$$
\left(\int_1^\infty\left|\text{Re}(\varphi_{Y_0})\right|^2\right)^{1/2}\leq
L\;.
$$
Inserting~(\ref{eq:upluspetitque1}) and~(\ref{eq:u-plusgrand1})
in~(\ref{eq:PhiXmod}), we get the result. 
\end{proof}

\begin{lemma}
\label{Lemma:Karamata}
There exists a constant $B>0$ such that, for all $u>0$, we have
\begin{equation*}
u\int_{u}^{\infty} \sqrt{\log\left(1+v\right)}\frac{\rmd v}{v^2} \leq B\sqrt{\log\left(1+u\right)}
\end{equation*}
\end{lemma}
\begin{proof}

For all $u>0$, we have
$$
u\int_{u}^{\infty} \sqrt{\log\left(1+v\right)}\frac{\rmd v}{v^2} =\int_{1}^{\infty}\sqrt{\log\left(1+uy\right)}\frac{\rmd y}{y^2} \leq \sqrt{u}\int_1^\infty \frac{\rmd y}{y^{3/2}}=2\sqrt{u}\;.
$$
As $u \to 0$, $\sqrt{\log\left(1+u\right)}$ is equivalent to $\sqrt{u}$.\\
As $u\to \infty$, the Karamata's Theorem (see \cite{resnick2007extreme}[Theorem 0.6]) applied to the function $u\to\sqrt{\log(1+u)}u^{-2}$, which is regularly varying with index $-2$, gives that
$$
u\int_{u}^{\infty} \sqrt{\log\left(1+v\right)}\frac{\rmd v}{v^2} \underset{u\to\infty}{\sim}\sqrt{\log\left(1+u\right)} \;,
$$
which concludes the proof.
\end{proof}
\newpage
\bibliographystyle{plain}
\bibliography{shotnoise}
\end{sloppypar}
\end{document}